\documentclass[11pt]{article}

\usepackage[a4paper,margin=1in]{geometry}
\usepackage{amsmath,amssymb,amsfonts,amsthm,mathtools}
\usepackage{bm}
\usepackage{array,booktabs,longtable}
\usepackage[hidelinks]{hyperref}
\usepackage{xcolor}
\usepackage[expansion=false]{microtype}
\pdfmapfile{+/usr/share/texlive/texmf-dist/fonts/map/dvips/amsfonts/cm.map}
\pdfmapfile{+/usr/share/texlive/texmf-dist/fonts/map/dvips/amsfonts/symbols.map}
\allowdisplaybreaks

\newtheorem{theorem}{Theorem}[section]

\newtheorem{corollary}[theorem]{Corollary}

\newtheorem{conjecture}[theorem]{Conjecture}
\newtheorem{remark}[theorem]{Remark}

\newcommand{\Acal}{\mathcal A}
\newcommand{\Bcal}{\mathcal B}
\newcommand{\Dcal}{\mathcal D}
\newcommand{\Rcal}{\mathcal R}
\newcommand{\Scal}{\mathcal S}

\newcommand{\Res}{\operatorname{Res}}
\newcommand{\Disc}{\operatorname{Disc}}
\newcommand{\Gal}{\operatorname{Gal}}
\newcommand{\thX}{\theta_X}
\newcommand{\thY}{\theta_Y}
\newcommand{\Th}{\Theta}

\title{Differential Recursions, Projective Discriminants,\\
and Algebraic Generating Functions}
\author{S.~Voloshyn\\
\small Bogolyubov Institute for Theoretical Physics,\quad
National Academy of Sciences of Ukraine\\[-1mm]
\small Kyiv, Ukraine}
\date{}

\begin{document}
\maketitle

\begin{abstract}
Two families of bivariate generating functions emerge from differential
recursions with rational initial data: a self-dual hierarchy and a
reduced-symmetry hierarchy. Their coefficient functions give the
palindromic and alternating triangles A140136 and A336110 at order
$p=3$. Starting from these recursions, we derive a decomposition into
$p-1$ Horn companions and an exact differential intertwiner between
one companion and the self-dual function. The Narayana endpoint and
the Gross--Witten--Wadia case remain explicit illustrations: at
$p=3$ a single quartic controls the full nonsymmetric generating
function and two edge resolvents with tetrahedral symmetry.

For every $p\geq3$ we prove that the distinguished reduced generating
function is algebraic of exact degree $2(p-1)$. Its rational
uniformization over $H=h^{p-1}$ lifts, in the original variable $h$,
to a Fermat curve $s^{p-1}+w^{p-1}=1$ of genus
$(p-2)(p-3)/2$. The generic monodromy and Galois group are
$S_{2(p-1)}$. The characteristic resultant carries a projective
$S_3$ action; its cubic contribution to the discriminant has a
local cusp explanation. Exact sextic and octic examples, a
degree-ten specialization, and boundary factorizations exhibit
the higher geometry. Finally, a uniform nonexact differential on
the Fermat curve, together with full $S_{2(p-1)}$ monodromy, proves
that the distinguished self-dual function is transcendental for every
$p\geq4$, settling the algebraicity question for both distinguished
hierarchies. The complete generic discriminant factorization remains open.
\end{abstract}

\section{Introduction}

The objects considered here were first encountered in a much more elementary form: as rows of rational functions generated by a differential recursion.  The denominators followed a rigid linear law in the recursion index, while the numerators formed symmetric or alternating integer triangles.  Only after several low-order cases had been computed did it become apparent that the coefficient arrays, the differential equations, the birational transformations of the arguments, and the algebraic equations of the corresponding bivariate generating functions were different manifestations of one geometry.
The original motivation came from the large-$N$ $SU(N)$ Gross--Witten--Wadia (GWW) integral and its charge-sector generating functions~\cite{GrossWitten1980,Wadia1980,BCV2020}. The recursions, rather than an assumed algebraic equation, remain our starting point.

There are two closely related families.  In the self-dual family the distinguished generating function satisfies
\begin{equation}
 x\,\partial_x^p L(x,h)=h\,\partial_h^p L(x,h),
 \label{eq:intro-selfdual}
\end{equation}
and admits an exchange of the two variables after the natural covering is introduced.  In the reduced family the $x$-operator is unchanged but the $h$-operator is shifted:
\begin{equation}
 x\,\partial_x^p F(x,h)
 =\frac{1}{(p-1)^2}h^{p-2}\partial_h^{p-1}
 \!\left(h^{3-p}\partial_hF(x,h)\right).
 \label{eq:intro-reduced}
\end{equation}
The reduced family has less symmetry in its lower-order terms.
Nevertheless, its algebraicity holds uniformly: for every integer
$p\geq3$, the distinguished reduced generating function is algebraic
over $\mathbb Q(x,h^{p-1})$ and has degree exactly $2(p-1)$.
The proof in Section~\ref{app:uniform-reduced-cover} resums the Horn
decomposition on a single rationally parametrized covering. It also
establishes the algebraicity of each individual Horn companion.

The purpose of this paper is twofold. First, we keep the elementary origin
of the problem visible: the first coefficient functions and the third-order
numerator triangles are displayed explicitly. Second, we develop the
higher structure suggested by these examples. The differential recursions
lead to a Horn decomposition, whose principal symbols define a projective
characteristic curve. A uniform algebraic covering then connects this
geometry with explicit equations and their discriminants.

Two low-order cases are especially important.  At $p=2$ one recovers the Narayana generating function, whose familiar square root becomes a projectively natural double cover.  At $p=3$ the two hierarchies meet in an exceptional way.  The reduced third-order problem splits into two algebraic sectors.  One is the self-dual sector and the other is a contiguous Horn companion.  Both are controlled by a single quartic root configuration; their six branches are the six edges of a tetrahedron, represented either by pair sums or by squared pair differences.  The lifted symmetry is therefore $S_4$.

The differential intertwiner relates one reduced companion to the
normalized self-dual series. Its inverse gives $p-3$ successive
quadratures. At every $p\geq4$ the first inverse step encounters a
nonexact differential on the Fermat curve. Full $2(p-1)$-sheet
monodromy prevents its cancellation among the residue-projected
branches, proving transcendence throughout the higher self-dual
hierarchy. This resolves the earlier open question for all its
distinguished members with $p\geq4$; we develop the elliptic $p=4$
case first as an example.

The uniform covering explains the previously observed degrees
$4,6,8,10$ at $p=3,4,5,6$ and gives the degree $2(p-1)$ at every
order. The explicit quartic, sextic and octic remain valuable: they
display the ramification and Galois structure that the degree theorem
alone does not determine. The generic Galois group at every order is $S_{2(p-1)}$.
The complete discriminant factorization remains conjectural.

The order of presentation follows the discovery and then the proof.
Section~\ref{sec:origin} begins with the recursions and coefficient
triangles. Sections~\ref{sec:p2} and \ref{sec:p3} retain the Narayana
endpoint and the full GWW quartic. Section~\ref{sec:horn-general}
derives the Horn decomposition and its bridge to the self-dual family.
Section~\ref{app:uniform-reduced-cover} proves the uniform covering
and exact algebraic degree, followed by the Fermat normalization and
full monodromy in Section~\ref{sec:fermat-monodromy}.
Sections~\ref{sec:projective-general}--\ref{sec:boundary} explain the
characteristic resultant, higher explicit examples, and discriminant
geometry. Section~\ref{sec:selfdual-higher} establishes
transcendence at every self-dual order $p\geq4$. The final sections state the
remaining conjecture and conclusions; appendices contain explicit
equations, computational checks and further boundary analysis.

\section{Differential recursions and coefficient triangles}
\label{sec:origin}

Throughout the paper we write
\begin{equation}
 q=p-1,
 \qquad H=h^q.
 \label{eq:qdef}
\end{equation}
The normalization used below is chosen so that the reduced generating function starts at order $H$ and the self-dual generating function starts at order $x^qH$.

\subsection{The reduced hierarchy}

Define
\begin{equation}
 F_p(x,h)=\sum_{k\ge0}R_k^{(p)}(x)\,h^{q(k+1)}.
 \label{eq:Freduced-series}
\end{equation}
The differential equation is Eq.~\eqref{eq:intro-reduced}.  Coefficient comparison gives
\begin{equation}
 x\frac{d^p}{dx^p}R_{k-1}^{(p)}(x)
 =\lambda^{\rm red}_{p,k}R_k^{(p)}(x),
 \label{eq:red-rec}
\end{equation}
where
\begin{equation}
 \lambda^{\rm red}_{p,k}
 =\frac{(qk+q)(qk+1)(qk)(qk-1)\cdots(qk-q+2)}{q^2}.
 \label{eq:red-lambda}
\end{equation}
The lowest coefficient is
\begin{equation}
 R_0^{(p)}(x)=\frac{q-1}{q}\frac{x}{1+x}.
 \label{eq:red-R0}
\end{equation}
The first two nontrivial descendants can be written in closed form:
\begin{equation}
 R_1^{(p)}(x)=
 (-1)^q\frac{q-1}{2}\frac{x}{(1+x)^{q+2}},
 \label{eq:red-R1}
\end{equation}
\begin{equation}
 R_2^{(p)}(x)=
 \frac{(q-1)q(q+1)}{6}
 \frac{x(1-x)}{(1+x)^{2q+3}}.
 \label{eq:red-R2}
\end{equation}
Thus the denominator exponent grows linearly with the recursion level.  More generally the rows may be written as a rational prefactor times a polynomial numerator, and the latter obeys a reciprocal palindromic or anti-palindromic relation.  This elementary fact is the origin of the surviving birational involution in the higher reduced hierarchy.

\subsection{The self-dual hierarchy}

The normalized self-dual family is
\begin{equation}
 L_p(x,h)=\sum_{k\ge0}S_k^{(p)}(x)\,h^{q(k+1)},
 \label{eq:Sseries}
\end{equation}
with
\begin{equation}
 x\partial_x^pL_p=h\partial_h^pL_p.
 \label{eq:SDPDE}
\end{equation}
The coefficient recursion is
\begin{equation}
 x\frac{d^p}{dx^p}S_{k-1}^{(p)}(x)
 =\lambda^{\rm sd}_{p,k}S_k^{(p)}(x),
 \label{eq:SDrec}
\end{equation}
where
\begin{equation}
 \lambda^{\rm sd}_{p,k}
 =(qk)(qk+1)\cdots(qk+q).
 \label{eq:SDlambda}
\end{equation}
The initial rational function is
\begin{equation}
 S_0^{(p)}(x)=\frac{x^q}{1+x^q}.
 \label{eq:SD-S0}
\end{equation}
For example,
\begin{equation}
 S_1^{(p)}(x)
 =\frac{(q-1)!}{(2q)!}\,
 x\frac{d^{q+1}}{dx^{q+1}}
 \frac{x^q}{1+x^q}.
 \label{eq:SD-S1-general}
\end{equation}
The general denominator has the form $(1+x^q)^{pk+1}$; after a simple sign change in the variable $x^q$ the numerator coefficients form symmetric row triangles.

\subsection{The third-order rows: two concrete triangles}
\label{subsec:p3-triangles}

The order $p=3$ is the first place where the relation between the two hierarchies becomes nontrivial.  Here $q=2$.

For the reduced family it is convenient to remove the common factor $1/2$ and write $\widehat R_k=2R_k^{(3)}$.  The first rational coefficient functions are
\begin{align}
 \widehat R_0(x)&=\frac{x}{1+x},\nonumber\\
 \widehat R_1(x)&=\frac{x}{(1+x)^4},\nonumber\\
 \widehat R_2(x)&=\frac{2x(1-x)}{(1+x)^7},\nonumber\\
 \widehat R_3(x)&=\frac{x(5-14x+5x^2)}{(1+x)^{10}},\nonumber\\
 \widehat R_4(x)&=\frac{x(14-74x+74x^2-14x^3)}{(1+x)^{13}},\nonumber\\
 \widehat R_5(x)&=\frac{x(42-352x+668x^2-352x^3+42x^4)}{(1+x)^{16}}.
 \label{eq:p3-red-first}
\end{align}
The numerator coefficients form the alternating triangle A336110:
\begin{equation}
\begin{array}{c|rrrrrr}
0&1\\
1&1\\
2&2&-2\\
3&5&-14&5\\
4&14&-74&74&-14\\
5&42&-352&668&-352&42
\end{array}
\label{eq:A336110}
\end{equation}
The first column is the Catalan sequence.  The alternating palindromicity of the rows is already visible before any bivariate generating function is introduced.

The self-dual third-order coefficients start instead with
\begin{align}
 S_0^{(3)}(x)&=\frac{x^2}{1+x^2},\nonumber\\
 S_1^{(3)}(x)&=\frac{x^2(x^2-1)}{(1+x^2)^4},\nonumber\\
 S_2^{(3)}(x)&=\frac{x^2(1-7x^2+7x^4-x^6)}{(1+x^2)^7},\nonumber\\
 S_3^{(3)}(x)&=\frac{x^2(x^{10}-20x^8+75x^6-75x^4+20x^2-1)}{(1+x^2)^{10}}.
 \label{eq:p3-sd-first}
\end{align}
If $E_k(z)$ denotes the positive numerator polynomial, then
\begin{equation}
 S_k^{(3)}(x)=(-1)^k\frac{x^2E_k(-x^2)}{(1+x^2)^{3k+1}},
 \label{eq:A140-form}
\end{equation}
and the first rows of A140136 are
\begin{equation}
\begin{array}{c|rrrrrr}
0&1\\
1&1&1\\
2&1&7&7&1\\
3&1&20&75&75&20&1
\end{array}
\label{eq:A140136}
\end{equation}
Thus the two third-order problems begin with visibly different rational triangles, even though they will later be shown to live in the same tetrahedral algebraic field.

A second, symmetric way of packaging the self-dual coefficients is the square array
\begin{equation}
 S(m,n)=
 \frac{(m+n+1)!(2m+2n+1)!}
 {(m+1)!(2m+1)!(n+1)!(2n+1)!},
 \qquad m,n\ge0,
 \label{eq:A111910}
\end{equation}
which is A111910 and has a classical lattice-walk interpretation~\cite{KrewerasNiederhausen1981}. Its bivariate series begins
\begin{equation}
 \Phi(X,Y)=XY+X^2Y+XY^2+X^3Y+5X^2Y^2+XY^3+\cdots.
 \label{eq:Phi-grounded}
\end{equation}
This is the point at which the rational row picture turns naturally into a two-variable algebraic problem.

\section{The Narayana endpoint $p=2$}
\label{sec:p2}

The Narayana case is the lowest nondegenerate projective member and provides a control example.  Let $N(X,Y)$ denote the bivariate Narayana generating function.  It satisfies
\begin{equation}
 N^2+(1+X+Y)N+XY=0.
 \label{eq:Narayana}
\end{equation}
Introduce the shifted fiber coordinate
\begin{equation}
 Z_2=2N+1+X+Y.
 \label{eq:Z2}
\end{equation}
With the projective coordinates
\begin{equation}
 (z_1,z_2,z_3)=(-1,X,Y)
 \label{eq:projcoords}
\end{equation}
and elementary symmetric functions
\begin{equation}
 \sigma_1=X+Y-1,
 \qquad
 \sigma_2=XY-X-Y,
 \qquad
 \sigma_3=-XY,
 \label{eq:sigmas}
\end{equation}
Eq.~\eqref{eq:Narayana} becomes
\begin{equation}
 Z_2^2=\sigma_1^2-4\sigma_2.
 \label{eq:Narayana-projective}
\end{equation}
The right-hand side is the universal characteristic resultant $\Acal_2$ introduced later.  The three projective coordinates are permuted by $S_3$, so the familiar Narayana square root is naturally the fiber coordinate of an $S_3$-symmetric double covering. Related $S_3$ symmetries have an earlier combinatorial history: Salas and Sokal established an $S_3\times D_4$ action on the family of Graham--Knuth--Patashnik recurrences \cite{SalasSokal2021}, and Maier studied an $S_3$ action on number triangles and their exponential generating functions, including generalized Narayana cases \cite{Maier2023}. Those transformations act on a family of triangles; Eq.~\eqref{eq:Narayana-projective} exhibits the projective action on the algebraic cover of the particular ordinary bivariate Narayana function used here.

The boundary $Y=0$ gives
\begin{equation}
 N(X,0)=0,
 \qquad
 P_2(N;X,0)=N(N+1+X),
 \label{eq:Narayana-boundary}
\end{equation}
which is the $q=1$ member of the general boundary pattern of Section~\ref{sec:boundary}.

\section{The symmetry-enhanced case $p=3$}
\label{sec:p3}

\subsection{Three operators and the exceptional symmetry}

Set
\begin{equation}
 R=u\partial_u^3,
 \qquad
 K=h\partial_h^3,
 \qquad
 E=u\partial_u+h\partial_h,
 \qquad
 \Delta=E(E-1)(E-2).
 \label{eq:p3operators}
\end{equation}
For the distinguished third-order solution one has
\begin{equation}
 RG=KG=\Delta G.
 \label{eq:p3triple}
\end{equation}
The equality of the lower-order Euler polynomials is special to $p=3$ and will be explained in Section~\ref{sec:projective-general}.  It is the differential origin of the full scalar $S_3$ symmetry.

Passing to the covering variables
\begin{equation}
 X=u^2,
 \qquad Y=h^2,
 \label{eq:p3XY}
\end{equation}
the base group is the projective permutation group of $(-1,X,Y)$.  The square-root deck group is $V_4=C_2\times C_2$, and the natural lifted extension is
\begin{equation}
 V_4\rtimes S_3\simeq S_4.
 \label{eq:S4extension}
\end{equation}

\subsection{The self-dual sextic}

Define the invariants
\begin{equation}
 I=\sigma_1^2-12\sigma_2,
 \label{eq:Idef}
\end{equation}
\begin{equation}
 J=\sigma_1^2\sigma_2-6\sigma_2^2+9\sigma_1\sigma_3,
 \label{eq:Jdef}
\end{equation}
and the squared Vandermonde
\begin{equation}
 \Dcal=(1+X)^2(1+Y)^2(X-Y)^2.
 \label{eq:Ddef}
\end{equation}
For the self-dual series \eqref{eq:Phi-grounded}, introduce
\begin{equation}
 Z_3=1+X+Y+2\Phi(X,Y).
 \label{eq:Z3def}
\end{equation}
Then
\begin{theorem}[Self-dual third-order algebraic equation]
The distinguished branch satisfies
\begin{equation}
 Z_3^6-IZ_3^4-8JZ_3^2-16\Dcal=0.
 \label{eq:p3sextic}
\end{equation}
Equivalently, $U=Z_3^2$ satisfies
\begin{equation}
 C(U)=U^3-IU^2-8JU-16\Dcal=0.
 \label{eq:p3cubic}
\end{equation}
\end{theorem}

On $Y=0$ one has the exact factorization
\begin{equation}
 P_3(Z_3;X,0)=
 \left[Z_3^2-(1+X)^2\right]
 \left[Z_3^2-4X\right]^2.
 \label{eq:p3boundary-sextic}
\end{equation}
The analytic generating-function branch is selected by $Z_3(X,0)=1+X$.

The cubic discriminant is
\begin{equation}
 \Disc_U C
 =256XY\left[(1-X-Y)^3-27XY\right]^3.
 \label{eq:p3disc}
\end{equation}
Thus the nontrivial characteristic curve is
\begin{equation}
 \Acal_3(X,Y)=(1-X-Y)^3-27XY.
 \label{eq:A3}
\end{equation}

\subsection{One quartic and the tetrahedral $S_4$ field}

Choose
\begin{equation}
 \delta=(1+X)(1+Y)(X-Y),
 \qquad \delta^2=\Dcal,
 \label{eq:delta}
\end{equation}
and consider the quartic
\begin{equation}
 Q_4(t)=t^4-\frac I2t^2+4\delta t+
 \left(\frac{I^2}{16}+2J\right).
 \label{eq:quartic}
\end{equation}
Let $t_1,t_2,t_3,t_4$ be its roots.  Since $\sum_i t_i=0$, the three quantities
\begin{equation}
 U_1=(t_1+t_2)^2,
 \quad
 U_2=(t_1+t_3)^2,
 \quad
 U_3=(t_1+t_4)^2
 \label{eq:Uroots}
\end{equation}
are the three roots of Eq.~\eqref{eq:p3cubic}.  Hence the six branches of the symmetric sextic can be written as
\begin{equation}
 Z_{ij}=t_i+t_j,\qquad 1\le i<j\le4.
 \label{eq:pair-sums}
\end{equation}
They are naturally the six edges of a tetrahedron with vertices $t_i$.

\subsection{The nonsymmetric Horn companion}

The second third-order sector is
\begin{equation}
 \Psi(X,Y)=\sum_{r,s\ge0}c_{r,s}X^rY^s,
 \label{eq:Psi}
\end{equation}
with
\begin{equation}
 c_{r,s}=\frac{(\frac12)_{r+s}(1)_{r+s}(\frac32)_{r+s}}
 {(1)_r(\frac32)_r(2)_r(\frac12)_s(1)_s(\frac32)_s}.
 \label{eq:Psi-coeff}
\end{equation}
It satisfies the Horn system
\begin{equation}
 \thX(\thX+\tfrac12)(\thX+1)\Psi
 =X(\Th+\tfrac12)(\Th+1)(\Th+\tfrac32)\Psi,
 \label{eq:Horn1}
\end{equation}
\begin{equation}
 (\thY-\tfrac12)\thY(\thY+\tfrac12)\Psi
 =Y(\Th+\tfrac12)(\Th+1)(\Th+\tfrac32)\Psi,
 \label{eq:Horn2}
\end{equation}
where $\Th=\thX+\thY$.

Define
\begin{equation}
 W=1-\frac X2\Psi,
 \qquad
 \Lambda=16YW^2.
 \label{eq:WLambda}
\end{equation}
Then the six branches are the squared edge differences
\begin{equation}
 \Lambda_{ij}=(t_i-t_j)^2.
 \label{eq:pair-diffs}
\end{equation}
Equivalently,
\begin{equation}
 \Rcal(\Lambda)=
 \Res_U\!\left(
 C(U),
 \Lambda^2-2(I-U)\Lambda+I^2+2IU-3U^2+32J
 \right)=0.
 \label{eq:companion-resultant}
\end{equation}
Expanding,
\begin{equation}
\begin{aligned}
\Rcal(\Lambda)={}&\Lambda^6-4I\Lambda^5
+2(3I^2+8J)\Lambda^4
-4(I^3+4IJ-104\Dcal)\Lambda^3\\
&+(I^4-16I^2J-448J^2-384I\Dcal)\Lambda^2\\
&+16(I^3J+24IJ^2+18I^2\Dcal+432J\Dcal)\Lambda
+\Disc_U C.
\end{aligned}
\label{eq:companion-sextic}
\end{equation}
The two algebraic sectors are directly related on a common edge by
\begin{equation}
 \Lambda=I-Z^2-\frac{8\delta}{Z}.
 \label{eq:sumdiff-relation}
\end{equation}
Thus the self-dual and nonsymmetric third-order generating functions are not separate accidents: they are two edge representations of one quartic $S_4$ field.

\subsection{The full nonsymmetric $p=3$ series: the quartic itself}
\label{sec:gww-full-quartic}

The preceding sextics describe edge coordinates, whereas the full generating
function in GWW normalization gives a vertex
coordinate.  Denote the latter by $\mathcal G(u,h)$ and use the covering
variables $X=u^2$, $Y=h^2$ from Eq.~\eqref{eq:p3XY}.  Its distinguished
formal branch begins
\begin{equation}
 \mathcal G(u,h)=\frac{u}{2}+\frac{h^2}{4(1+u)}
 -\frac{u h^4}{16(1+u)^4}
 +\frac{u(u-1)h^6}{32(1+u)^7}+O(h^8).
 \label{eq:gww-full-series}
\end{equation}
Unlike the individual Horn companion $\Psi$, this full nonsymmetric series
combines the two third-order residue sectors with the elementary terms
required by its normalization.  The same invariant quartic
\eqref{eq:quartic} governs it: with
\begin{equation}
 \quad \tau=-4\mathcal G-\frac{1+X-Y}{2},\qquad Q_4(\tau)=0.\quad
 \label{eq:gww-vertex}
\end{equation}
This is an invertible affine change of fiber coordinate, so the generic
degree of $\mathcal G$ is four.  In particular, the monic polynomial
$P(\mathcal G;X,Y)=Q_4(-4\mathcal G-(1+X-Y)/2)/256$ has
\begin{equation}
 \begin{aligned}
 P(\mathcal G;X,0)
 &=\frac{(4\mathcal G^2-X)(4\mathcal G+1+X)^2}{64},\\
 \operatorname{Disc}_{\mathcal G}P
 &=\frac{XY}{2^{16}}\,\Acal_3(X,Y)^3.
 \end{aligned}
 \label{eq:gww-vertex-boundary-disc}
\end{equation}
The first identity selects the branch $\mathcal G(u,0)=u/2$; the second
identifies precisely the same characteristic factor as in the edge
resolvents.  The quartic is generically irreducible: at $u=h=1$ its
integer-normalized specialization reduces modulo $3$ to
$g^4+2g^3+g^2+1$, which is irreducible over $\mathbb F_3$.
Thus the fiber degrees of the reduced hierarchy start with Narayana's
$2$ at $p=2$, this full nonsymmetric $4$ at $p=3$, and the degrees
$6,8,10$ at $p=4,5,6$.  The visible third-order sextics are resolvents
of the quartic and do not interrupt the sequence $2(p-1)$.

\subsection{The odd Horn companion as a rational numerator triangle}
\label{sec:gww-odd-numerator}

The quartic vertex generating function has a useful parity split.
Put $t=h^2/4$ and $z=-u$, and use the self-dual series
$\Phi$ of Eq.~\eqref{eq:Phi-grounded} together with the Horn
companion $\Psi$ of Eq.~\eqref{eq:Psi}.  In the normalization of
Eq.~\eqref{eq:gww-full-series},
\begin{equation}\label{eq:gww-parity-decomposition}
 \mathcal G(-z,2\sqrt t)+\frac z2
   =t+\frac14\Phi(z^2,4t)+zt\Psi(4t,z^2).
\end{equation}
The coefficient of $t^s$ in $\Psi(4t,x)$ has a rational form
\begin{equation}\label{eq:gww-odd-rational-rows}
 \begin{aligned}
 \Psi(4t,x)&=\sum_{s\ge0}t^s\frac{D_s(x)}{(1-x)^{3s+1}},\\
 D_0&=1,\qquad D_1=1+6x+x^2,\\
 D_2&=2+56x+140x^2+56x^3+2x^4.
 \end{aligned}
\end{equation}
The rational denominator and the two properties of $D_s$ hold at
every order. Indeed, for fixed $s$,
\begin{equation}
 4^sc_{s,r}=
 \frac{4^s(r+\frac12)_s(r+1)_s(r+\frac32)_s}
 {(1)_s(\frac32)_s(2)_s}=:P_s(r).
 \label{eq:gww-odd-polynomial-coeff}
\end{equation}
This is a polynomial of degree $3s$ in $r$, so its ordinary generating
function has denominator $(1-x)^{3s+1}$. Moreover,
$P_s(-r-s-1)=(-1)^{3s}P_s(r)$ and $P_s(-1)=\cdots=P_s(-s)=0$.
The rational-series reciprocity identity
$\sum_{r\geq0}P_s(r)x^{-r}=-\sum_{r\geq1}P_s(-r)x^r$
(understood as an identity of rational functions) gives
\begin{equation}
 [t^s]\Psi(4t,x^{-1})=(-1)^{s+1}x^{s+1}[t^s]\Psi(4t,x),
 \qquad D_s(x)=x^{2s}D_s(x^{-1}).
 \label{eq:gww-odd-reciprocity}
\end{equation}
The constant coefficient $4^sc_{s,0}=C_s>0$ proves that the
degree of $D_s$ is exactly $2s$.

Positivity also follows directly from the Horn coefficients. Put
$f_s(x)=D_s(x)/(1-x)^{3s+1}$ and $E=x\partial_x$. Equation
\eqref{eq:gww-odd-polynomial-coeff} yields
\begin{equation}
 f_{s+1}=
 \frac{4(E+s+\frac12)(E+s+1)(E+s+\frac32)}
 {(s+1)(s+\frac32)(s+2)}f_s.
 \label{eq:gww-odd-positive-recursion}
\end{equation}
If $f(x)=D(x)/(1-x)^m$, $D(x)=\sum_kd_kx^k$, then $(E+a)f$
has denominator $(1-x)^{m+1}$ and numerator coefficient
\begin{equation}
 [x^k](1-x)^{m+1}(E+a)f
 =(k+a)d_k+(m-k+1-a)d_{k-1}.
 \label{eq:gww-odd-positive-step}
\end{equation}
Starting at $D_0=1$, apply this rule in order with
$(m,a)=(3s+1,s+\frac12),(3s+2,s+1),(3s+3,s+\frac32)$.
The numerator degrees after these steps are $2s+1$, $2s+1$
and $2s+2$, respectively. Both multipliers in
\eqref{eq:gww-odd-positive-step} are nonnegative wherever
the preceding numerator has a nonzero coefficient, and at least
one contribution is strictly positive for every coefficient in
the final support. Induction proves that all coefficients of
$D_s$ are positive.

Thus $D_s$ is a positive, palindromic numerator triangle distinct
from A336110 and A140136. Let $B_s(z)$ be the unsigned version
of the A336110 row ($B_0=1$), and write $E_s(x)$ for the A140136
row, defined by
$\sum_{m\ge0}S(m,s)x^m=E_s(x)/(1-x)^{3s+1}$, with $S$
given in Eq.~\eqref{eq:A111910}. Coefficient extraction from
Eq.~\eqref{eq:gww-parity-decomposition} gives the exact identity
\begin{equation}\label{eq:gww-odd-even-identity}
 \ (1+z)^{3s+1}B_s(z)
       =D_s(z^2)+4^s z E_s(z^2).\ 
\end{equation}
In particular, the even powers on the left determine $D_s$;
its row sum is $D_s(1)=2^{3s}\mathrm{A000309}(s)$.

There is also a conjectural refinement of the description trees
of type $(2,2)$ counted by A000309~\cite{CoriSchaeffer2003}.
If $\nu$ is the number of
leaves plus the number of maximal internal labels, then the
claim $B_s(z)=\sum_{\tau}z^{\nu(\tau)-2}$ agrees with exact
enumeration through $s=8$.  Assuming this refinement,
$D_s$ counts the same trees decorated by subsets of
$3s+1$ slots for which $\nu-2+|M|$ is even.  A proof of this
conjectural tree refinement and a bijection to the known walk
model remain to be found; Bernardi~\cite{Bernardi2007} gives a
bijection between Kreweras walks and loopless triangulations.
The exact algebraic formulas above
do not depend on this interpretation.

Finally, the sextic Eq.~\eqref{eq:companion-sextic} applies to
$\Lambda=16Y(1-X\Psi/2)^2$, not directly to $\Psi$.
Substitution gives a degree-twelve equation for $\Psi$ itself;
the exact rational specialization $(X,Y)=(2,3)$ verifies its generic
irreducibility.  This distinction is useful when comparing the
six unoriented edge coordinates with the twelve oriented ones.

\section{General Horn decomposition and the bridge between the hierarchies}
\label{sec:horn-general}

We now return to arbitrary $p\ge3$ and write $q=p-1$.

\subsection{Residue-class decomposition of the reduced hierarchy}

Introduce the Horn variables
\begin{equation}
 \widehat X=(-1)^q x^q,
 \qquad
 \widehat Y=(-1)^q q^2H.
 \label{eq:HornXY}
\end{equation}
For $r=1,\dots,q$, define
\begin{equation}
 \Phi_r(\widehat X,\widehat Y)
 =\sum_{m,n\ge0}c_{m,n}^{(r)}\widehat X^m\widehat Y^n,
 \label{eq:Phi-r}
\end{equation}
where
\begin{equation}
 c_{m,n}^{(r)}=
 \frac{
 \displaystyle\prod_{j=0}^{q}\left(\frac{r+j}{q}\right)_{m+n}}
 {
 \displaystyle\prod_{j=0}^{q}\left(\frac{r+j}{q}\right)_m
 (2)_n\prod_{j=2}^{q+1}\left(\frac jq\right)_n
 }.
 \label{eq:general-Horn-coeff}
\end{equation}

\begin{theorem}[Horn decomposition]
The distinguished reduced generating function admits the formal decomposition
\begin{equation}
 F_p(x,h)=\frac{q-1}{q}H
 \sum_{r=1}^{q}(-1)^{r-1}x^r
 \Phi_r(\widehat X,\widehat Y).
 \label{eq:general-Horn-decomp}
\end{equation}
\end{theorem}

\begin{proof}[Proof sketch]
Expand the rational coefficients at $x=0$ and split every exponent as $N=qm+r$, $1\le r\le q$.  The recursion produces a product of $p=q+1$ consecutive affine factors in $N$ and in the $H$-index.  Applying the multiplication formula for the Gamma function to each residue class gives Eq.~\eqref{eq:general-Horn-coeff}.  Summing over the $q$ residue classes yields Eq.~\eqref{eq:general-Horn-decomp}.
\end{proof}

Each $\Phi_r$ satisfies the Horn system
\begin{equation}
 \prod_{j=0}^{q}
 \left(\thX+\frac{r+j}{q}-1\right)\Phi_r
 =\widehat X
 \prod_{j=0}^{q}
 \left(\Th+\frac{r+j}{q}\right)\Phi_r,
 \label{eq:general-Horn-X}
\end{equation}
\begin{equation}
 \left[\prod_{b\in B_q}(\thY+b-1)\right]\Phi_r
 =\widehat Y
 \prod_{j=0}^{q}
 \left(\Th+\frac{r+j}{q}\right)\Phi_r,
 \label{eq:general-Horn-Y}
\end{equation}
where
\begin{equation}
 B_q=\left\{2,\frac2q,\frac3q,\ldots,\frac{q+1}{q}\right\}.
 \label{eq:Bq}
\end{equation}

\subsection{The self-dual series and an exact differential intertwiner}

Define the normalized self-dual Horn series
\begin{equation}
 \Scal_p(X,Y)=\sum_{m,n\ge0}s_{m,n}^{(p)}X^mY^n,
 \label{eq:S-horn}
\end{equation}
with
\begin{equation}
 s_{m,n}^{(p)}=
 \frac{
 \displaystyle\prod_{j=0}^{q}\left(1+\frac jq\right)_{m+n}}
 {
 \displaystyle\prod_{j=0}^{q}\left(1+\frac jq\right)_m
 \prod_{j=0}^{q}\left(1+\frac jq\right)_n
 }.
 \label{eq:S-horn-coeff}
\end{equation}
The last reduced companion, $r=q$, has exactly the same numerator and $X$-denominator parameters.  Their coefficient ratio is
\begin{equation}
 \frac{c_{m,n}^{(q)}}{s_{m,n}^{(p)}}
 =\prod_{j=2}^{q-1}\frac{qn+j}{j}.
 \label{eq:bridge-ratio}
\end{equation}

\begin{theorem}[Reduced/self-dual bridge]
Let $\theta_Y=Y\partial_Y$.  Then
\begin{equation}
 \Phi_q(X,Y)=\Dcal_p\Scal_p(X,Y),
 \qquad
 \Dcal_p=
 \prod_{j=2}^{p-2}\frac{(p-1)\theta_Y+j}{j}.
 \label{eq:bridge}
\end{equation}
The order of $\Dcal_p$ is $p-3$.
\end{theorem}

The first cases are
\begin{equation}
 \Dcal_3=1,
 \qquad
 \Dcal_4=\frac{3\theta_Y+2}{2},
 \label{eq:D34}
\end{equation}
\begin{equation}
 \Dcal_5=\frac{(4\theta_Y+2)(4\theta_Y+3)}{6},
 \qquad
 \Dcal_6=\frac{(5\theta_Y+2)(5\theta_Y+3)(5\theta_Y+4)}{24}.
 \label{eq:D56}
\end{equation}
This theorem gives a precise meaning to the statement that the two hierarchies are contiguous rather than independent.  At $p=3$ the self-dual function is literally one reduced companion.  For $p=4$ it is one integration step above an algebraic companion; for $p=5$ and $p=6$ two and three inverse differential steps are required.

\begin{corollary}[Iterated quadratures for the self-dual series]
\label{cor:selfdual-quadratures}
For every $p\geq3$, put $q=p-1$ and define the inverse Euler operator
\begin{equation}
 (\mathcal T_j f)(X,Y)
 =\frac{j}{q}Y^{-j/q}\int_0^Y
       t^{j/q-1}f(X,t)\,dt
 =\sum_{n\geq0}\frac{j}{qn+j}[Y^n]f(X,Y)\,Y^n,
 \qquad 2\leq j\leq q-1.
 \label{eq:selfdual-inverse-Euler}
\end{equation}
The integral is understood coefficientwise for formal power series.
Then
\begin{equation}
 \quad
 \Scal_p(X,Y)=
   (\mathcal T_2\circ\cdots\circ\mathcal T_{q-1})
       \Phi_q(X,Y).\quad
 \label{eq:selfdual-iterated-quadratures}
\end{equation}
The empty composition at $p=3$ is the identity.  Since
$\Phi_q$ is algebraic by Corollary~\ref{cor:uniform-horn-algebraic},
every normalized self-dual series is obtained from an algebraic
function by $p-3$ successive quadratures.  This statement does not
decide whether the resulting series is algebraic.
\end{corollary}

\begin{proof}
The coefficient formula in Eq.~\eqref{eq:selfdual-inverse-Euler}
shows that
$\mathcal T_j\bigl((q\theta_Y+j)f/j\bigr)=f$.
All these Euler operators commute and have nonzero eigenvalues
$(qn+j)/j$ on $Y^n$ for $n\geq0$.  Applying their inverses to
Eq.~\eqref{eq:bridge} gives
Eq.~\eqref{eq:selfdual-iterated-quadratures}, including its
normalization at $Y=0$.
\end{proof}

\section{A uniform algebraic covering of the reduced hierarchy}
\label{app:uniform-reduced-cover}

The Horn decomposition gives an algebraic parametrization of the entire
reduced hierarchy, without reconstructing a separate polynomial at each
order.  In this section $q=p-1\geq 2$, $H=h^q$, and
$\varepsilon=(-1)^q$.  The function $F_p(x,h)$ has the normalization of
Eq.~\eqref{eq:general-Horn-decomp}.

\begin{theorem}[Uniform covering and exact algebraic degree]
\label{thm:uniform-reduced-cover}
There is a unique element $s=s(x,H)\in\mathbb Q(x)[[H]]$ with
$s(x,0)=1$ satisfying
\begin{equation}
 (1-s^q)(x+s)^q=\varepsilon q^2H\,s^q.
 \label{eq:uniform-s-cover}
\end{equation}
For this branch, the distinguished reduced generating function is
\begin{equation}
 \begin{aligned}
 F_p(x,h)=\frac{\varepsilon}{q^3}\bigg[&(1+x)^q
 +(q-1)\left(1+\frac{x}{s}\right)^q\\
 &-q\left(1+\frac{x}{s}\right)^{q-1}
 -qx(x+s)^{q-1}\bigg].
 \end{aligned}
 \label{eq:uniform-F-cover}
\end{equation}
In particular, $F_p$ is algebraic over $\mathbb Q(x,H)$ and its
minimal polynomial has degree exactly $2q=2(p-1)$.
\end{theorem}

\begin{proof}
Write $Y=\varepsilon q^2H$ and
\begin{equation}
 U(x,Y)=\frac{qF_p(x,h)}{(q-1)H}.
 \label{eq:uniform-U-def}
\end{equation}
This quotient is a formal power series in $H$.  For completeness,
its coefficients can also be obtained directly from the defining
recursion.  Write $r_n(N)=[x^N]R_n^{(p)}(x)$ for $N\geq1$.
Equation~\eqref{eq:red-rec}, with $p=q+1$, gives
\begin{equation}
 \lambda^{\rm red}_{p,n}r_n(N)
 =(N+q)_{\underline{q+1}}\,r_{n-1}(N+q),\qquad
 r_0(N)=\frac{q-1}{q}(-1)^{N-1}.
 \label{eq:uniform-direct-recursion}
\end{equation}
Here the underline denotes a falling factorial.  From
Eq.~\eqref{eq:red-lambda},
\begin{equation}
 \prod_{j=1}^n\lambda^{\rm red}_{p,j}
 =\frac{(qn+1)!(n+1)!}{q^n},\qquad
 \prod_{j=1}^n(N+qj)_{\underline{q+1}}
 =q^n(N)_{qn}(N/q+1)_n.
 \label{eq:uniform-direct-products}
\end{equation}
Consequently,
\begin{equation}
 r_n(N)=\frac{q-1}{q}(-1)^{N-1}
 (\varepsilon q^2)^n
 \frac{(N)_{qn}(N/q+1)_n}{(qn+1)!(n+1)!}.
 \label{eq:uniform-direct-coefficient}
\end{equation}
This establishes the next identity directly from the recursion;
the same expression follows by putting $N=qm+r$ in
Eqs.~\eqref{eq:general-Horn-coeff} and
\eqref{eq:general-Horn-decomp}:
\begin{equation}
 U(x,Y)=\sum_{N\geq1}\sum_{n\geq0}
 (-1)^{N-1}x^NY^n
 \frac{(N)_{qn}(N/q+1)_n}{(qn+1)!(n+1)!}.
 \label{eq:uniform-double-series}
\end{equation}
Indeed, the product of Pochhammer symbols in the numerator of
Eq.~\eqref{eq:general-Horn-coeff}, after division by its
$m$-dependent denominator, is
$q^{-qn}(N)_{qn}(N/q+1)_n$; the remaining denominator is
$q^{-qn}(qn+1)!(n+1)!$.  The powers of $(-1)^q x^q$ and the
residue prefactor give $(-1)^{N-1}x^N$.

Let $\theta_Y=Y\partial_Y$, and apply
$(q\theta_Y+1)(\theta_Y+1)$ to $U$.  The two factors cancel
the corresponding linear denominators in
Eq.~\eqref{eq:uniform-double-series}.  Using
\begin{equation}
 \frac{(N)_{qn}}{(qn)!}=[u^{qn}](1-u)^{-N},\qquad
 \frac{(N/q+1)_n}{n!}=[v^n](1-v)^{-N/q-1},
 \label{eq:uniform-two-coeff-extractions}
\end{equation}
the sum over $N$ becomes geometric:
\begin{align}
 V(x,Y)&:=(q\theta_Y+1)(\theta_Y+1)U(x,Y)\notag\\
 &=\sum_{n\geq0}Y^n[u^{qn}v^n]
 \frac{x}{(1-v)[(1-u)(1-v)^{1/q}+x]}\notag\\
 &=\sum_{n\geq0}Y^n[v^n]
 \frac{x(1-v)^{n-1}}
 {\bigl(x+(1-v)^{1/q}\bigr)^{qn+1}}.
 \label{eq:uniform-V-extract}
\end{align}
Define
\begin{equation}
 a(v)=\frac{x}{(1-v)\bigl(x+(1-v)^{1/q}\bigr)},\qquad
 b(v)=\frac{1-v}{\bigl(x+(1-v)^{1/q}\bigr)^q}.
 \label{eq:uniform-a-b}
\end{equation}
The final expression in Eq.~\eqref{eq:uniform-V-extract} is
$\sum_{n\geq0}Y^n[v^n]a(v)b(v)^n$.  The formal
Lagrange--residue identity states that, if $w=Yb(w)$ and $w(0)=0$,
then
\begin{equation}
 \sum_{n\geq0}Y^n[v^n]a(v)b(v)^n
 =\frac{a(w)}{1-Yb'(w)}.
 \label{eq:uniform-Lagrange-residue}
\end{equation}
Put $s=(1-w)^{1/q}$, with the branch $s(0)=1$.  The equation
$w=Yb(w)$ becomes Eq.~\eqref{eq:uniform-s-cover}.  Moreover,
$b'(w)=-x/(x+s)^{q+1}$, so Eq.~\eqref{eq:uniform-Lagrange-residue}
simplifies to
\begin{equation}
 V(x,Y)=\frac{x}{x+s^{q+1}}.
 \label{eq:uniform-V-rational}
\end{equation}

Differentiating Eq.~\eqref{eq:uniform-s-cover} shows directly that
\begin{equation}
 (q\theta_Y+1)\frac{x}{x+s}
 =\frac{x}{x+s^{q+1}}.
 \label{eq:uniform-first-inversion}
\end{equation}
Both $x/(x+s)$ and $(\theta_Y+1)U$ have constant coefficient
$x/(1+x)$ at $Y=0$.  Since $q\theta_Y+1$ is invertible on
$\mathbb Q(x)[[Y]]$, Eqs.~\eqref{eq:uniform-V-rational} and
\eqref{eq:uniform-first-inversion} imply
\begin{equation}
 (\theta_Y+1)U=\frac{x}{x+s}.
 \label{eq:uniform-second-inversion}
\end{equation}
As $Y$ is proportional to $H$, Eq.~\eqref{eq:uniform-U-def}
therefore gives the useful differential identity
\begin{equation}
 \quad\partial_H F_p(x,h)
 =\frac{q-1}{q}\frac{x}{x+s}.\quad
 \label{eq:uniform-H-derivative}
\end{equation}

It remains to integrate this algebraic derivative.  Solving
Eq.~\eqref{eq:uniform-s-cover} for $H$ and differentiating,
\begin{equation}
 H(s)=\frac{(1-s^q)(x+s)^q}{\varepsilon q^2s^q},\qquad
 \frac{dH}{ds}=
 -\frac{(x+s)^{q-1}(x+s^{q+1})}
 {\varepsilon q\,s^{q+1}}.
 \label{eq:uniform-H-map}
\end{equation}
Consequently,
\begin{equation}
 \frac{dF_p}{ds}=
 -\frac{(q-1)x}{\varepsilon q^2}
 (x+s)^{q-2}\left(1+\frac{x}{s^{q+1}}\right).
 \label{eq:uniform-exact-differential}
\end{equation}
For integer $q\geq2$, the Laurent expansion of the right-hand
side contains no $s^{-1}$ term: its negative powers range from
$s^{-q-1}$ to $s^{-3}$.  Thus it has a rational primitive in
$\mathbb Q(x,s)$.  Integrating from $s=1$, where $F_p=0$, and
simplifying gives Eq.~\eqref{eq:uniform-F-cover}.  This proves
algebraicity for every $q\geq2$.

We finally prove that the degree does not drop.  Work over
$k=\mathbb Q(x)$ and set $K=k(s)$.  Here $s$ is transcendental over
$k$: its formal expansion at $H=0$ begins
$s=1-\varepsilon q(1+x)^{-q}H+O(H^2)$, and a nonconstant element
of $k[[H]]$ cannot be algebraic over the constant field $k$.
The rational map $s\mapsto H(s)$
in Eq.~\eqref{eq:uniform-H-map} has poles of order $q$ at $s=0$
and $s=\infty$, and hence degree $2q$.  Thus $[K:k(H)]=2q$.
The rational map $s\mapsto F_p(s)$ in
Eq.~\eqref{eq:uniform-F-cover} has poles of order $q$ at $s=0$
and $q-1$ at $s=\infty$, with nonzero leading coefficients in
$k=\mathbb Q(x)$.  Thus $[K:k(F_p)]=2q-1$.  Since
$k(H,F_p)\subseteq K$, the degree $[K:k(H,F_p)]$ divides both
$2q$ and $2q-1$.  Their greatest common divisor is one, so
$K=k(H,F_p)$ and $[k(H,F_p):k(H)]=2q$, as claimed.
\end{proof}

The theorem gives an unexpanded polynomial at every order:
one may take the field norm
\begin{equation}
 P_p(x,H,Z)=
 \operatorname{Norm}_{\mathbb Q(x,s)/\mathbb Q(x,H)}
 \bigl(Z-F_p(x,s)\bigr),
 \label{eq:uniform-norm-polynomial}
\end{equation}
with $H(s)$ and $F_p(x,s)$ given above.  The degree argument shows
that this monic polynomial is minimal; equivalently it can be
calculated by a resultant after clearing the powers of $s$ in
Eq.~\eqref{eq:uniform-F-cover}.

\begin{corollary}[Algebraicity of every Horn companion]
\label{cor:uniform-horn-algebraic}
For each $r=1,\ldots,q$, the individual Horn series
$\Phi_r(\widehat X,\widehat Y)$ of
Eq.~\eqref{eq:general-Horn-coeff} is algebraic over
$\mathbb Q(\widehat X,\widehat Y)$.
\end{corollary}

\begin{proof}
Let $\zeta$ be a primitive $q$th root of unity, and define
the finite residue projection
\begin{equation}
 A_r(x,H)=\frac1q\sum_{j=0}^{q-1}
 \zeta^{-jr}F_p(\zeta^j x,h).
 \label{eq:uniform-horn-projection}
\end{equation}
The Horn decomposition, Eq.~\eqref{eq:general-Horn-decomp},
implies
\begin{equation}
 A_r(x,H)=\frac{q-1}{q}H(-1)^{r-1}x^r
 \Phi_r(\varepsilon x^q,\varepsilon q^2H).
 \label{eq:uniform-horn-projected-sector}
\end{equation}
For each $j$, Theorem~\ref{thm:uniform-reduced-cover} makes
$F_p(\zeta^j x,h)$ algebraic over
$\mathbb Q(\zeta)(x,H)$, so their finite linear combination
$A_r$ is algebraic over that field.  Dividing by the nonzero
rational prefactor in Eq.~\eqref{eq:uniform-horn-projected-sector}
preserves algebraicity.  Finally,
$\mathbb Q(\zeta)(x,H)$ is an algebraic extension of
$\mathbb Q(\widehat X,\widehat Y)$ under
$\widehat X=\varepsilon x^q$ and
$\widehat Y=\varepsilon q^2H$.  Transitivity of algebraicity
proves the assertion.
\end{proof}

This corollary concerns the reduced Horn companions.  The
self-dual series $\Scal_p$ is related to $\Phi_q$ by the
differential intertwiner \eqref{eq:bridge}; recovering
$\Scal_p$ involves integration when $p>3$, and the
corollary makes no assertion that this integral is algebraic.

\section{Fermat fibers and full monodromy}
\label{sec:fermat-monodromy}
\begin{theorem}[Fermat normalization]\label{thm:fermat}
Let $K=\mathbb C(x)$, with $x$ transcendental. Then the field $K(h,F_p)$ is the function field of the smooth projective Fermat curve
\begin{equation}\label{eq:fermat}
 C_q:\qquad s^q+w^q=1,
\end{equation}
where $\varepsilon=(-1)^q$, $a_q^q=\varepsilon q^2$, $w=a_qhs/(x+s)$, and
\begin{equation}\label{eq:cover-map}
 h=\frac{(x+s)w}{a_qs},\qquad F_p=F_p(x,s).
\end{equation}
The map $h:C_q\to\mathbb P^1_h$ has degree $2q$, and the smooth projective curve has genus
\begin{equation}\label{eq:genus}
 g(C_q)=\frac{(q-1)(q-2)}2=\frac{(p-2)(p-3)}2.
\end{equation}
\end{theorem}
\begin{proof}
Equation (\ref{eq:uniform-s-cover}) gives $w^q=1-s^q$, and (\ref{eq:cover-map}) gives the converse. Since $s\in K(H,F_p)$ by Theorem~\ref{thm:uniform-reduced-cover} and $w=a_qhs/(x+s)$, the two displayed fields coincide over $K$. After a generic specialization of $x$, $h$ has $q$ simple poles above $s=0$ and another $q$ above $s=\infty$, giving degree $2q$. The projective plane curve $S^q+W^q=Z^q$ is smooth in characteristic zero; the genus formula follows from the plane curve genus formula, or by Riemann--Hurwitz for its projection to the $s$ line.
\end{proof}

At $p=3$, $C_2$ is rational. At $p=4$, $C_3$ is elliptic. At $p=5$, $C_4$ is a smooth plane quartic of genus three and is nonhyperelliptic. For $\xi^q=1$, the deck action $(s,w)\mapsto(s,\xi w)$ sends $h\mapsto\xi h$ and fixes $H$ and $F_p$; quotienting by it recovers the rational $H$ covering. For fixed generic $x$, the abstract curve stays the same while the map to $\mathbb P^1_h$ changes. Its interior critical points satisfy
\begin{equation}\label{eq:critical}
 \frac{dh}{ds}=-\frac{x+s^{q+1}}{a_qs^2w^{q-1}}=0.
\end{equation}
The characteristic discriminant therefore describes exceptional fibers of this projection; it is not the Fermat curve itself. The characteristic resultant of Section~\ref{sec:projective-general} reappears among these branch values.

\begin{theorem}[Full generic monodromy]\label{thm:full-monodromy}
For every $q\geq2$, the geometric monodromy group of $h:C_q\to\mathbb P^1_h$ is $S_{2q}$ for generic $x$. Consequently the geometric Galois group of the degree-$2q$ minimal polynomial of $F_p(x,H)$ over $\mathbb C(x,H)$ is also $S_{2q}$.
\end{theorem}
\begin{proof}
Away from $s=0,\infty,-x$ and $w=0$, equation (\ref{eq:critical}) gives $q+1$ roots of $s^{q+1}=-x$, each with $q$ distinct lifts $w$; all are simple critical points for generic $x$. At $s=0$ and $s=\infty$ the function $h$ has simple poles, at $s=-x$ it has a simple zero, and at $w=0$ the local parameter $w$ makes $h$ a simple zero. There is no additional ramification. Indeed Riemann--Hurwitz and (\ref{eq:genus}) give the total ramification
\begin{equation}\label{eq:rh-total}
 2g(C_q)-2+2\deg(h)=q(q+1),
\end{equation}
exactly the number just found.

We must also check that the simple critical points have \emph{distinct branch values}; Riemann--Hurwitz alone would not suffice. Write $x=-t^{q+1}$ and $s=t\xi$, where $\xi^{q+1}=1$. At a critical point, equation (\ref{eq:uniform-s-cover}) simplifies to
\begin{equation}\label{eq:critical-H}
 H_{\xi}=\frac{\varepsilon}{q^2}(1-t^q\xi^{-1})^{q+1}
 =\frac{\varepsilon}{q^2}\bigl(1-(q+1)t^q\xi^{-1}+O(t^{2q})\bigr).
\end{equation}
For distinct $\xi$, the first nonconstant terms differ. Thus the $q+1$ values $H_{\xi}$ are distinct and nonzero outside a proper algebraic subset of $x$. Each has $q$ distinct $h$th-root values coming from the $q$ lifts in $w$, and different $H_{\xi}$ cannot share an $h$ value. Every branch cycle of the connected degree-$2q$ cover is therefore a transposition. A transitive permutation group generated by transpositions is the full symmetric group $S_{2q}$.

The rational $s\mapsto H$ cover has degree $2q$. Base change along $H=h^q$ yields the irreducible curve $C_q$, also of degree $2q$ over $h$; its monodromy is a subgroup of the monodromy of $s\mapsto H$. Since the subgroup is already $S_{2q}$, the original monodromy is $S_{2q}$ as well. The identity $K(H,F_p)=K(s)$ from Theorem~\ref{thm:uniform-reduced-cover} transfers this result to the minimal polynomial of $F_p$.
\end{proof}

\section{Universal characteristic resultant and projective symmetry}
\label{sec:projective-general}

\subsection{Principal symbols}

The principal symbols of Eqs.~\eqref{eq:general-Horn-X}--\eqref{eq:general-Horn-Y} do not depend on the residue label $r$:
\begin{equation}
 \xi^p=\widehat X(\xi+\eta)^p,
 \qquad
 \eta^p=\widehat Y(\xi+\eta)^p.
 \label{eq:principalsymbol}
\end{equation}
Putting $a=\xi/(\xi+\eta)$ gives a universal elimination problem.  In the physical sign convention define
\begin{equation}
 \Acal_p(X,Y)=
 \Res_a\!\left(
 a^p-\varepsilon_pX,
 (1-a)^p-\varepsilon_pY
 \right),
 \qquad
 \varepsilon_p=(-1)^{p-1}.
 \label{eq:Ap-resultant}
\end{equation}
The curve $\Acal_p=0$ is rationally parametrized by
\begin{equation}
 X=\varepsilon_pa^p,
 \qquad
 Y=\varepsilon_p(1-a)^p.
 \label{eq:Aparam}
\end{equation}
Homogeneously, writing $a=s/(s+t)$,
\begin{equation}
 [z_1:z_2:z_3]
 =[-(s+t)^p:\varepsilon_ps^p:\varepsilon_pt^p].
 \label{eq:homog-param}
\end{equation}
The six anharmonic transformations
\begin{equation}
 a,\quad1-a,\quad a^{-1},\quad(1-a)^{-1},\quad
 \frac{a}{a-1},\quad\frac{a-1}{a}
 \label{eq:anharmonic}
\end{equation}
permute the three linear forms $s,t,-(s+t)$ and therefore induce the full projective group $S_3$ on $(-1,X,Y)$.

\subsection{Invariant forms for low orders}

Using Eq.~\eqref{eq:sigmas}, the first characteristic covariants are
\begin{align}
 \Acal_2={}&\sigma_1^2-4\sigma_2,\nonumber\\
 \Acal_3={}&-\sigma_1^3+27\sigma_3,\nonumber\\
 \Acal_4={}&(\sigma_1^2-4\sigma_2)^2-128\sigma_1\sigma_3,\nonumber\\
 \Acal_5={}&-\sigma_1^5+625\sigma_1^2\sigma_3-3125\sigma_2\sigma_3,\nonumber\\
 \Acal_6={}&(\sigma_1^2-4\sigma_2)^3
 -2754\sigma_1^3\sigma_3
 -12312\sigma_1\sigma_2\sigma_3
 +185193\sigma_3^2.
 \label{eq:A2A6}
\end{align}
Thus the projective characteristic geometry persists even when the scalar lower-order symmetry does not.

\subsection{Why $p=3$ is symmetry-enhanced}

The lower-order difference is particularly transparent in Euler form.  Define
\begin{equation}
 A_p(t)=\prod_{j=0}^{p-1}(t-j),
 \qquad
 B_p(t)=t\prod_{j=p-2}^{2p-4}(t-j).
 \label{eq:ABpoly}
\end{equation}
After the natural rescaling of $h$, the $x$-operator uses $A_p(\theta_x)$, while the reduced $h$-operator uses $B_p(\theta_h)$.  At $p=3$,
\begin{equation}
 A_3(t)=B_3(t)=t(t-1)(t-2),
 \label{eq:ABp3}
\end{equation}
and the three scalar operators can be permuted by $S_3$. For $p\ge4$
the root sets differ, so this enhancement of the \emph{reduced}
system is lost. This does not remove the exact $S_3$ action on the
self-dual Horn system, as the following direct comparison shows.
The common leading term $t^p$ still gives both systems the projective
$S_3$ characteristic curve $\Acal_p$.

Set $q=p-1$, $\Theta=\theta_X+\theta_Y$, and
\begin{equation}
 P_q(t)=\prod_{j=0}^{q}\left(t+\frac jq\right),\qquad
 Q_q(t)=(t+1)\prod_{j=2}^{q+1}\left(t+\frac jq-1\right).
 \label{eq:symmetry-PQ}
\end{equation}
The normalized self-dual system and the last reduced companion have,
respectively, the forms
\begin{align}
 P_q(\theta_X)\Scal_p&=X P_q(\Theta+1)\Scal_p,&
 P_q(\theta_Y)\Scal_p&=Y P_q(\Theta+1)\Scal_p,\notag\\
 P_q(\theta_X)\Phi_q&=X P_q(\Theta+1)\Phi_q,&
 Q_q(\theta_Y)\Phi_q&=Y P_q(\Theta+1)\Phi_q.
 \label{eq:symmetry-two-systems}
\end{align}
Both systems are preserved on their \emph{spaces of local solutions}
by the reciprocal transformation
\begin{equation}
 (\mathcal R_p f)(X,Y)=X^{-2}
 f\!\left(\frac1X,\frac{(-1)^pY}{X}\right).
 \label{eq:symmetry-reciprocal}
\end{equation}
Indeed $P_q(-1-t)=(-1)^pP_q(t)$, and substitution in
Eq.~\eqref{eq:symmetry-two-systems} verifies preservation of both
operator ideals. The self-dual system is also preserved by
$\mathcal S f(X,Y)=f(Y,X)$. The operators $\mathcal S$ and $\mathcal R_p$
satisfy $\mathcal S^2=\mathcal R_p^2=(\mathcal S\mathcal R_p)^3=1$,
so they realize $S_3$ on the self-dual solution space. For $p\ge4$,
$P_q\ne Q_q$, and exchanging $X,Y$ does not preserve the displayed
reduced system. Its exhibited argument symmetry is therefore the
involution $\mathcal R_p$; maximality among all birational
transformations is not asserted. At $p=3$ the two polynomials coincide,
and $\Phi_2=\Scal_3$.

This comparison concerns differential systems. The reciprocal map
typically sends the normalized germ at the origin to a different local
solution; the symmetry group of one selected branch may be smaller.
For $p=4$, $P_3=t(t+\frac13)(t+\frac23)(t+1)$ whereas
$Q_3=t(t-\frac13)(t+\frac13)(t+1)$. Correspondingly, the
coefficient of $Y$ in $\Scal_4$ is $1$, but in $\Phi_3$ it is
$5/2$, while the coefficient of $X$ is $1$ in both.

At the level of the \emph{collection} of residue components there is
also a finite symmetry: deck rotation $x\mapsto\zeta x$, $\zeta^q=1$,
and reciprocal reflection act on the $q$ components as a dihedral
group $D_q$. For $q=3$ this is $D_3\simeq S_3$, but it does not supply
the missing $X\leftrightarrow Y$ symmetry of the individual reduced
companion $\Phi_3$. The $p=3$ quartic has, in addition, the distinct
tetrahedral $S_4$ action on its four algebraic roots.

\section{Higher algebraic members of the reduced hierarchy}
\label{sec:higher-reduced}

Section~\ref{app:uniform-reduced-cover} proves algebraicity of the
reduced hierarchy and the exact degree $2(p-1)$ for every $p\geq3$.
The explicit equations below provide information beyond this uniform
result: their coefficients, discriminants, and Galois groups. The
$p=4$ and $p=5$ equations were reconstructed from exact series and
verified by implicit differentiation modulo the resulting polynomials;
at $p=6$ we give an exact degree-ten specialization.

\subsection{$p=4$: an exact sextic}

Set $H=h^3$.  The reduced fourth-order function satisfies a monic sextic
\begin{equation}
 P_4(x,H,F)=0,
 \qquad \deg_F P_4=6.
 \label{eq:p4-poly}
\end{equation}
The full coefficients are listed in Appendix~\ref{app:p4}.  Its boundary factorization is
\begin{equation}
 P_4(x,0,F)=
 \frac{F\,[27F+(1+x)^3]^3
 [27F^2+9x(1+x)F+x^2(x^2+x+1)]}{3^{12}}.
 \label{eq:p4boundary}
\end{equation}
The discriminant is
\begin{equation}
 \Disc_F P_4
 =3^{-64}H^4x^6
 \Acal_4(x^3,9H)^3\Bcal_4(x^3,H)^2,
 \label{eq:p4disc}
\end{equation}
where
\begin{align}
\Bcal_4(X,H)={}&(X-1)(X+1)^4
+18H(11X^4-260X^3+482X^2-260X+11)\nonumber\\
&+81H^2(145X^3+1809X^2-1809X-145)\nonumber\\
&+5832H^3(19X^2-387X+19)
+5458752H^4(1-X)+30233088H^5.
 \label{eq:B4}
\end{align}
A specialization at $x=H=1$ has Galois group $S_6$, and the good-specialization argument gives the generic group
\begin{equation}
 \Gal(P_4/\mathbb Q(x,H))\simeq S_6.
 \label{eq:p4gal}
\end{equation}
The factor $\Acal_4^3$ is the projectively symmetric characteristic ramification.  The factor $\Bcal_4^2$ corresponds to ordinary lower-order collisions; locally the former has cusp-type splitting $y^2\sim\varepsilon^3$, whereas the latter is nodal.

\subsection{$p=5$: an exact octic}

For $H=h^4$ the reduced fifth-order function satisfies an exact monic octic
\begin{equation}
 P_5(x,H,F)=0,
 \qquad \deg_F P_5=8,
 \label{eq:p5poly}
\end{equation}
verified by exact implicit differentiation.  Its characteristic covariant is
\begin{equation}
 \Acal_5
 =-\sigma_1^5+625\sigma_1^2\sigma_3-3125\sigma_2\sigma_3.
 \label{eq:A5again}
\end{equation}
The discriminant has the form
\begin{equation}
 \Disc_F P_5
 =-2^{-240}x^{12}H^9
 \Acal_5(x^4,16H)^3\Bcal_5(x^4,H)^2.
 \label{eq:p5disc}
\end{equation}
Here $\deg_H\Bcal_5=13$.  The full octic and $\Bcal_5$ are supplied as ancillary symbolic data; the formulas are too long to illuminate the main geometry in print.

At $x=H=1$ the specialization is irreducible modulo a good prime and exhibits Frobenius cycle types forcing
\begin{equation}
 \Gal(P_5/\mathbb Q(x,H))\simeq S_8.
 \label{eq:p5gal}
\end{equation}

\subsection{$p=6$: degree ten and the $S_{10}$ specialization}

For $H=h^5$ an exact specialization at $x=1$ satisfies an irreducible degree-ten equation
\begin{equation}
 P_{6,1}(H,F)=0,
 \qquad \deg_F P_{6,1}=10.
 \label{eq:p6decic}
\end{equation}
Its boundary is
\begin{equation}
 P_{6,1}(0,F)
 =\frac{F(125F+32)^5(125F^2+75F+11)^2}{5^{21}}.
 \label{eq:p6boundary}
\end{equation}
The characteristic covariant is
\begin{equation}
 \Acal_6=(\sigma_1^2-4\sigma_2)^3
 -2754\sigma_1^3\sigma_3
 -12312\sigma_1\sigma_2\sigma_3
 +185193\sigma_3^2.
 \label{eq:A6again}
\end{equation}
Independent generic slices reproduce
\begin{equation}
 \Disc_F P_6
 \doteq H^{16}\Acal_6(x^5,25H)^3\Bcal_6(x^5,H)^2,
 \qquad \deg_H\Bcal_6=24.
 \label{eq:p6disc-slices}
\end{equation}
At $x=H=1$, modular factorization gives a transitive degree-ten group containing a $7$-cycle and an odd $10$-cycle; Jordan's theorem then yields
\begin{equation}
 \Gal(P_{6,1}(1,F)/\mathbb Q)\simeq S_{10}.
 \label{eq:p6gal}
\end{equation}
This specialization illustrates the full generic group already proved
in Theorem~\ref{thm:full-monodromy}.

\subsection{The emerging sequence}

The exact and specialized data may be summarized as
\begin{equation}
\begin{array}{c|c|c|c}
 p&\text{fiber degree}&\text{generic Galois group}&\text{discriminant core}\\ \hline
2&2&S_2&\Acal_2\\
3&4\ \text{(quartic field)}&S_4&\Acal_3^3\\
4&6&S_6&\Acal_4^3\Bcal_4^2\\
5&8&S_8&\Acal_5^3\Bcal_5^2\\
6&10&S_{10}&\Acal_6^3\Bcal_6^2
\end{array}
\label{eq:degree-table}
\end{equation}
The third-order sextics are resolvents of the quartic field rather than minimal equations for a single preferred primitive element, which is why the visible sextic degree there should not be confused with the field degree four.

\section{Universal boundary factorization}
\label{sec:boundary}

The higher equations become particularly transparent on the boundary $H=0$.  Let again $q=p-1$ and define
\begin{equation}
 \beta_q(x)=\frac{(-1)^q}{q^3}(1+x)^q,
 \label{eq:betaq}
\end{equation}
\begin{equation}
 \beta_{q,\zeta}(x)=
 \frac{(-1)^q}{q^3}\left[(1+x)^q-(1+\zeta x)^q\right],
 \qquad \zeta^q=1.
 \label{eq:betazeta}
\end{equation}
Notice that $\beta_{q,1}=0$.

Corollary~\ref{cor:uniform-boundary}, proved from the uniform covering
in Section~\ref{app:uniform-reduced-cover}, gives the boundary
factorization for every $p\geq3$:
\begin{equation}
 P_p(x,0,F)=
 F\,[F-\beta_q(x)]^q
 \prod_{\substack{\zeta^q=1\\\zeta\ne1}}
 [F-\beta_{q,\zeta}(x)].
 \label{eq:boundary-general}
\end{equation}
Equivalently, the residual factor is
\begin{equation}
 \Rcal_q(x,F)=q^{-3(q-1)}
 \Res_z\!\left(
 \frac{z^q-1}{z-1},
 q^3F-(-1)^q[(1+x)^q-(1+zx)^q]
 \right).
 \label{eq:boundary-resultant}
\end{equation}
The $p=2$ Narayana case was established separately in
Section~\ref{sec:p2}. Its degree count is immediate:
\begin{equation}
 1+q+(q-1)=2q=2(p-1).
 \label{eq:degree-count}
\end{equation}
Thus the observed sequence $2,4,6,8,10$ has a direct boundary-sheet interpretation.

There are also two immediate discriminant consequences.  As $x\to0$,
\begin{equation}
 \beta_{q,\zeta}(x)
 =\frac{(-1)^q}{q^2}(1-\zeta)x+O(x^2).
 \label{eq:beta-smallx}
\end{equation}
The $q$ deck-labelled roots therefore coalesce linearly, contributing
\begin{equation}
 x^{2\binom q2}=x^{q(q-1)}=x^{(p-1)(p-2)}
 \label{eq:xdisc-power}
\end{equation}
to the discriminant.

The reduced $H$-operator has a fractional indicial exponent
\begin{equation}
 \lambda_0=\frac{q-1}{q}.
 \label{eq:indicial}
\end{equation}
The $q$-fold boundary root $F=\beta_q(x)$ therefore splits generically as
\begin{equation}
 F_j(H)=\beta_q(x)+c_jH^{(q-1)/q}+\cdots,
 \qquad j=1,\ldots,q,
 \label{eq:Hsplit}
\end{equation}
which contributes
\begin{equation}
 H^{2\binom q2(q-1)/q}=H^{(q-1)^2}=H^{(p-2)^2}.
 \label{eq:Hdisc-power}
\end{equation}
This explains the powers $H^4,H^9,H^{16}$ in the exact discriminants for $p=4,5,6$.

Combining the boundary geometry with the universal characteristic curve suggests the master factorization
\begin{equation}
 \Disc_F P_p
 \doteq
 x^{(p-1)(p-2)}H^{(p-2)^2}
 \Acal_p\!\left(x^{p-1},(p-1)^2H\right)^3
 \Bcal_p(x^{p-1},H)^2.
 \label{eq:master-disc}
\end{equation}
For $p=3$ the lower-order factor disappears: $\Bcal_3=1$.  The known degrees
\begin{equation}
 \deg_H\Bcal_4=5,
 \qquad
 \deg_H\Bcal_5=13,
 \qquad
 \deg_H\Bcal_6=24
 \label{eq:Bdegrees}
\end{equation}
suggest
\begin{equation}
 \deg_H\Bcal_p=\frac{(p-3)(3p-2)}{2}.
 \label{eq:Bdegree-conj}
\end{equation}
The vanishing at $p=3$ matches exactly the symmetry-enhancement condition $A_3=B_3$.

\section{Transcendence of the higher self-dual hierarchy}
\label{sec:selfdual-higher}

The reduced hierarchy is algebraic at every order $p\geq3$. For the self-dual function the bridge gives quadratures (Corollary~\ref{cor:selfdual-quadratures}). We first examine the lowest higher order $p=4$; its coefficient recursion is
\begin{equation}
 x\frac{d^4}{dx^4}S_{k-1}^{(4)}
 =(3k)(3k+1)(3k+2)(3k+3)S_k^{(4)},
 \label{eq:sd-p4-rec}
\end{equation}
with $S_0^{(4)}=x^3/(1+x^3)$.  The first numerator rows are
\begin{equation}
\begin{array}{c|rrrrrrrrr}
0&2\\
1&-5&17&-5\\
2&8&-192&772&-772&192&-8\\
3&-11&858&-12155&52832&-84801&52832&-12155&858&-11
\end{array}
\label{eq:p4sd-triangle}
\end{equation}
with alternating reciprocal symmetry.

The differential bridge now supplies the correct way to test algebraicity. The last reduced companion satisfies
\begin{equation}
 \Phi_3(X,Y)=\frac{3\theta_Y+2}{2}\Scal_4(X,Y).
 \label{eq:p4bridge}
\end{equation}
Hence
\begin{equation}
 \Scal_4(X,Y)=\frac23Y^{-2/3}
 \int_0^Y t^{-1/3}\Phi_3(X,t)\,dt.
 \label{eq:p4primitive}
\end{equation}
Since $\Phi_3$ is algebraic as a deck projection of the reduced sextic, this is an algebraic-primitive problem. The following theorem resolves it for the distinguished branch.

\begin{theorem}[First higher self-dual transcendence]\label{thm:sd4-transcendence}
The distinguished normalized self-dual generating function $\Scal_4(X,Y)$ is transcendental over $\mathbb C(X,Y)$.
\end{theorem}
\begin{proof}
Write $F_j(h)=F_4(\varpi^jx,h^3)$ for $j=0,1,2$, taking the distinguished branch of each reduced series at $H=0$. The residue projection (\ref{eq:uniform-horn-projected-sector}) and the bridge (\ref{eq:p4bridge}) yield
\begin{equation}\label{eq:sd4-primitive-bridge}
 d\bigl(h^2\Scal_4(-x^3,-9h^3)\bigr)
 =\frac{1}{x^3}\sum_{j=0}^{2}\frac{F_j(h)}{h^2}\,dh.
\end{equation}
Because $F_j/h$ is algebraic, integration by parts reduces the algebraicity of $\Scal_4$ to the existence of an algebraic primitive of $\omega=\sum_{j=0}^2\eta_j$, where $\eta_j=dF_j/h$. A differential possessing an algebraic primitive in a finite extension is already exact in the original function field after taking the normalized field trace.

First, one constituent is not exact. On $E=C_3:s^3+w^3=1$, with $h=(x+s)w/(a_3s)$, direct differentiation gives
\begin{equation}\label{eq:sd4-elliptic-reduction}
 \frac{dF_4}{ds}=\frac{2x}{9}(x+s)\left(1+\frac{x}{s^4}\right),\qquad
 \frac{dF_4}{h}=\frac{2a_3x}{9}s\frac{ds}{w}
 -\frac{a_3x^2}{9}d\left(\frac{w^2}{s^2}\right).
\end{equation}
The differential $\alpha=s\,ds/w$ is not exact. It has double poles only at the three points at infinity; a primitive could have at most simple poles there and no other poles. Riemann--Roch gives $L(D)=\langle1,s,w\rangle$ for their sum $D$. But $d(A+Bs+Cw)=(B-Cs^2/w^2)\,ds$ cannot equal $s\,ds/w$. Thus $\eta_0$ is nonexact for $x\ne0$.

To exclude cancellation, the substitutions $s_j=\varpi^js$ identify all three curves over $h$, and
\begin{equation}\label{eq:sd4-rotation}
 F_4(\varpi^jx,\varpi^js)-F_4(x,s)
 =-\frac{(1+\varpi^jx)^3-(1+x)^3}{27},
\end{equation}
which is constant in $h$. In the $s$ coordinate, these three branches approach $1,\varpi^{-1},\varpi^{-2}$ as $H\to0$; they are distinct members of the six conjugates above $h$. Theorem~\ref{thm:full-monodromy} establishes their full $S_6$ monodromy for generic $x$. As a concrete check at $x=2$, the four critical values of $H=h^3$ are roots of the squarefree polynomial
\begin{equation}\label{eq:sd4-critical-poly}
729H^4-2268H^3+12438H^2+4900H+729.
\end{equation}
Each produces three distinct nonzero $h$ values. Theorem~\ref{thm:full-monodromy}, rather than a check at this single $x$, supplies the generic assertion.

Let $L$ be the Galois closure of $K(E)/K(h)$ and let $\eta_1,\ldots,\eta_6$ denote the conjugates as differentials in $\Omega^1_{L/K}$. The three distinguished indices form a subset $I$ and $\omega=\sum_{i\in I}\eta_i$ in this field. If $\omega$ had an algebraic primitive, taking a normalized trace from a finite compositum would make $\omega$ exact already in $L$. An $S_6$ transposition exchanging one index in $I$ with one outside $I$ would then make their difference exact. All pairwise differences would be exact. Newton sums for $(s^3-1)(x+s)^3=9h^3s^3$ give
\begin{equation}\label{eq:sd4-trace}
 \sum_{i=1}^{6}F_i=2h^3-\frac{x^3+6x^2+6x+1}{9},
 \qquad \sum_{i=1}^{6}\eta_i=\frac{d\sum_iF_i}{h}=d(3h^2).
\end{equation}
Exactness of every difference and the total sum would make each $\eta_i$ exact in $L$. A normalized field trace back to $K(E)$ then makes the original constituent exact, contradicting (\ref{eq:sd4-elliptic-reduction}). Hence $\omega$ has no algebraic primitive, and (\ref{eq:sd4-primitive-bridge}) proves the claim. Transcendence over $\mathbb C(x,h)$ implies transcendence over the subfield $\mathbb C(X,Y)$ under $X=-x^3$, $Y=-9h^3$.
\end{proof}

The elliptic calculation is the first instance of a uniform obstruction.
Inverting the bridge once always gives a differential on the Fermat
curve with the same character under a root-of-unity rotation.

\begin{theorem}[Transcendence of the higher self-dual hierarchy]
\label{thm:selfdual-all-transcendence}
For every integer $p\geq4$, the distinguished normalized self-dual
generating function $\Scal_p(X,Y)$ of Eq.~\eqref{eq:S-horn} is
transcendental over $\mathbb C(X,Y)$.
\end{theorem}

\begin{proof}
Put $q=p-1\geq3$, $\varepsilon=(-1)^q$, $H=h^q$, and choose
$a_q$ with $a_q^q=\varepsilon q^2$. Work at generic $x$ and set
$X=\varepsilon x^q$, $Y=\varepsilon q^2h^q$. Since
$q\theta_Y=h\partial_h$, the bridge \eqref{eq:bridge} yields
\begin{equation}
 \Phi_q(X,Y)=\frac{1}{(q-1)!h}
 \frac{d^{q-2}}{dh^{q-2}}
 \left[h^{q-1}\Scal_p(X,Y)\right].
 \label{eq:selfdual-uniform-bridge-h}
\end{equation}
Let $\zeta$ be a primitive $q$th root of unity and write
$F_j(h)=F_p(\zeta^jx,h)$ for $0\leq j<q$. The residue
projection \eqref{eq:uniform-horn-projected-sector} at $r=q$
then gives
\begin{equation}
 \frac{d^{q-2}}{dh^{q-2}}
 \left[h^{q-1}\Scal_p(\varepsilon x^q,
                      \varepsilon q^2h^q)\right]
 =\frac{(-1)^{q-1}(q-2)!}{x^q h^{q-1}}
   \sum_{j=0}^{q-1}F_j(h).
 \label{eq:selfdual-uniform-branch-sum}
\end{equation}
If $\Scal_p$ were algebraic, the expression on the left and
its derivatives would be algebraic over $\mathbb C(x,h)$.
Integrating \eqref{eq:selfdual-uniform-branch-sum} once, with
\begin{equation}
 \int\frac{F_j(h)}{h^{q-1}}\,dh
 =-\frac{F_j(h)}{(q-2)h^{q-2}}
  +\frac{1}{q-2}\int\frac{dF_j(h)}{h^{q-2}},
 \label{eq:selfdual-uniform-integration-parts}
\end{equation}
would therefore give an algebraic primitive of the sum of the
$q$ selected differentials $\eta_j=dF_j/h^{q-2}$.

First, one constituent has no algebraic primitive. On the Fermat
normalization $C_q:s^q+w^q=1$, where
$h=(x+s)w/(a_qs)$, Eq.~\eqref{eq:uniform-exact-differential}
gives
\begin{equation}
 \eta_0=-\frac{(q-1)x a_q^{q-2}}{\varepsilon q^2}
 \left(\frac{s^{q-2}\,ds}{w^{q-2}}
 +\frac{x\,ds}{s^3w^{q-2}}\right),
 \qquad
 \frac{ds}{s^3w^{q-2}}
 =-\frac12d\left(\frac{w^2}{s^2}\right).
 \label{eq:selfdual-uniform-differential}
\end{equation}
Thus exactness of $\eta_0$ would imply exactness of
$\beta_q=s^{q-2}ds/w^{q-2}$. This differential is regular
away from the $q$ points at infinity, where it has double poles.
Any algebraic primitive, after normalized trace to the function
field of $C_q$, would have at most simple poles at the infinity
divisor $D_\infty$. As $C_q$ is a smooth plane curve of degree
$q\geq3$ and $D_\infty$ is a line section,
\begin{equation}
 H^0(C_q,\mathcal O(D_\infty))=\langle1,s,w\rangle.
 \label{eq:selfdual-uniform-pole-space}
\end{equation}
Under $w\mapsto\zeta w$, the form $\beta_q$ has character
$\zeta^2$. The derivatives of the functions in
\eqref{eq:selfdual-uniform-pole-space} have only characters
$1$ and $\zeta$. These are distinct from $\zeta^2$ for
$q\geq3$. Therefore $\beta_q$ and $\eta_0$ are nonexact.

To exclude cancellation, replacing $(x,s)$ by
$(\zeta^jx,\zeta^js)$ preserves $H(s)$, while
Eq.~\eqref{eq:uniform-F-cover} gives
\begin{equation}
 F_p(\zeta^jx,\zeta^js)-F_p(x,s)
 =\frac{\varepsilon}{q^3}
  \left[(1+\zeta^jx)^q-(1+x)^q\right],
 \label{eq:selfdual-uniform-rotation}
\end{equation}
which is constant in $h$. More explicitly, let $s_j(H)$ be the
distinguished solution of the covering equation with $x$ replaced
by $\zeta^jx$, so that $s_j(0)=1$, and put
$t_j(H)=\zeta^{-j}s_j(H)$. Writing $H_x(s)$ for the rational
map in Eq.~\eqref{eq:uniform-H-map} with parameter $x$, we have
$H_x(t_j)=H_{\zeta^jx}(s_j)=H$ and $t_j(0)=\zeta^{-j}$.
Equation~\eqref{eq:selfdual-uniform-rotation} gives
$F_j(h)=F_p(x,t_j(H))+c_j(x)$, where $c_j(x)$ is independent
of $h$. After differentiating, this constant disappears.
Consequently the $q$ forms $\eta_j$ are distinct conjugates
among the $2q$ conjugates of $\eta_0$ over $\mathbb C(x,h)$:
their $s$ coordinates approach
$1,\zeta^{-1},\ldots,\zeta^{-(q-1)}$ at $H=0$.
Let $L$ be the Galois closure of the degree-$2q$ covering and
index these selected forms by a subset $I$ of size $q$.
Theorem~\ref{thm:full-monodromy} gives the full group
$S_{2q}$ on the conjugates.

The sum of all $2q$ forms is exact. Indeed, the rational trace
of $F_p(x,s)$ under $s\mapsto H(s)$ has no finite $H$-poles:
the only poles of $F_p(x,s)$ on the $s$ sphere lie over
$H=\infty$. Its growth there is at most linear in $H$.
Near $s=0$ each of the $q$ local branches has
$F_p/H\to(q-1)/q$; the branches near $s=\infty$ have
$F_p/H\to0$. Hence for some $B_q(x)\in\mathbb C(x)$,
\begin{equation}
 \sum_{i=1}^{2q}F_i=(q-1)H+B_q(x),\qquad
 \sum_{i=1}^{2q}\eta_i
 =d\left(\frac{q(q-1)}2h^2\right).
 \label{eq:selfdual-uniform-trace}
\end{equation}
If $\sum_{i\in I}\eta_i$ had an algebraic primitive, a
normalized trace from a finite extension would make it exact
in $L$. A transposition exchanging an index in $I$ with one
outside $I$ would then make the difference of those two
forms exact. Since $S_{2q}$ contains every transposition,
all pairwise differences $\eta_i-\eta_j$ would be exact.
The exact total sum in \eqref{eq:selfdual-uniform-trace}
would force every $\eta_i$ to be exact in $L$. Tracing back
to the single-curve field contradicts nonexactness of $\eta_0$.
Therefore the selected sum has no algebraic primitive,
contradicting \eqref{eq:selfdual-uniform-integration-parts}.
Transcendence after the algebraic substitution $(X,Y)=
(\varepsilon x^q,\varepsilon q^2h^q)$ implies transcendence
over $\mathbb C(X,Y)$.
\end{proof}

At $p=3$ the bridge has order zero, and the putative obstruction
reduces to $\beta_2=ds$, which is exact. Thus the exceptional
algebraicity of the self-dual third-order member is consistent
with both the bridge and the pole-character calculation.

There is nevertheless a distinguished algebraic blow-up on the reciprocal mirror.  With
\begin{equation}
 X=-1+\varepsilon,
 \qquad
 H=\varepsilon^4t,
 \end{equation}
and an appropriate rescaling of the self-dual function, the limit is
\begin{equation}
 M(t)=-t\,{}_4F_3\!\left(
 \begin{matrix}1,\frac14,\frac12,\frac34\\[1mm]
 \frac43,\frac53,2\end{matrix};-256t\right),
 \label{eq:mirror-hypergeom}
\end{equation}
which satisfies a quartic algebraic equation.  The factor $1+256t$ in its discriminant is precisely the blow-up limit of $\Acal_4$.  This algebraic mirror reduction is compatible with the generic transcendence theorem: algebraicity can return on a special blow-up even though the two-variable distinguished function is transcendental.

\section{General conjectures and open problems}
\label{sec:conjectures}

The calculations above suggest the following general picture.

The uniform covering theorem settles algebraicity and the exact degree
of the distinguished reduced generating function for every $p\geq3$.
Theorem~\ref{thm:full-monodromy} determines the generic Galois group
at every order. The remaining general conjecture concerns the
complete discriminant factorization.
Theorem~\ref{thm:selfdual-all-transcendence} also settles the former
algebraicity question for the distinguished self-dual series: its
members are transcendental for all $p\geq4$. This theorem does not
classify every other local solution or every special slice.

\begin{conjecture}[Universal discriminant]
The generic discriminant is given by Eq.~\eqref{eq:master-disc}, with $\deg_H\Bcal_p$ given by Eq.~\eqref{eq:Bdegree-conj}.
\end{conjecture}

The main structural results do not depend on this conjecture: the
Horn decomposition, the differential bridge, and the projective
characteristic curve are valid at the level of the differential
recursions themselves.

Several problems remain.  The local $(2,3)$ cusp calculation in
Section~\ref{app:uniform-reduced-cover} explains the cubic contribution
from a generic characteristic point.  The complete discriminant
factorization still requires a uniform account of the remaining singularities,
their possible collisions, and the lower-order factor $\Bcal_p$.
An invariant-theoretic description of $\Bcal_p$ would also be valuable.
For the self-dual hierarchy, Corollary~\ref{cor:selfdual-quadratures}
gives a finite tower of quadratures. Theorem~\ref{thm:selfdual-all-transcendence}
proves that its distinguished function is transcendental for every
$p\geq4$. The first inverse step already detects the universal
nonexact differential; the structure and possible simplifications
of later inverse steps remain to be studied.

\section{Conclusions}

The differential recursions studied here begin with elementary rational functions and integer numerator triangles, but lead to a surprisingly rigid algebraic geometry.  The reduced hierarchy decomposes into $p-1$ Horn companions, one of which is connected to the self-dual hierarchy by a differential operator of order $p-3$.  Their common principal symbol produces a rational projective curve $\Acal_p=0$ carrying a universal $S_3$ action.

The two low orders are symmetry-enhanced.  Narayana's quadratic is the $p=2$ double cover $Z_2^2=\Acal_2$.  At $p=3$ the lower-order Euler polynomials coincide, the scalar $S_3$ symmetry is restored, and the lifted geometry becomes tetrahedral.  The two visibly different third-order numerator triangles A336110 and A140136 lead to two algebraic edge representations of the same $S_4$ quartic field.

Beyond the symmetry-enhanced third order, the reduced hierarchy remains
algebraic at every order $p\geq3$. The uniform rational covering proves
that its distinguished generating function has degree $2(p-1)$ and
establishes the boundary factorization. The explicit sextic and octic
at $p=4,5$ and the degree-ten $p=6$ specialization display the
resulting Galois and discriminant structure. Their common characteristic factor
$\Acal_p^3$ is explained locally by cusp geometry,
while a uniform proof of the complete discriminant formula remains open.

Thus the coefficient triangles, Horn systems, projective characteristic
curves, algebraic coverings and differential intertwiners belong to one
hierarchy. The reduced algebraicity and its exact degree are established.
The Fermat normalization and full symmetric monodromy show why the
quadratures on the self-dual side cannot be algebraic at any $p\geq4$:
a nonexact differential on one sheet cannot cancel in the selected
sum of conjugates. Thus the previously open algebraicity question
for the distinguished higher self-dual members is resolved. A uniform proof of the complete generic
discriminant factorization remains open.

\appendix

\section{The exact reduced $p=4$ sextic}
\label{app:p4}

For $H=h^3$ the monic polynomial is
\begin{equation}
 P_4=F^6+A_5F^5+A_4F^4+A_3F^3+A_2F^2+A_1F+A_0,
\end{equation}
where
\begin{equation}
 A_5=\frac{(1+x)(x^2+5x+1)}9-2H,
\end{equation}
\begin{align}
A_4={}&\frac{x^6+15x^5+60x^4+83x^3+60x^2+15x+1}{243}\nonumber\\
&-\frac{2}{27}(3x^3+15x^2+15x+2)H+\frac43H^2,
\end{align}
\begin{align}
A_3={}&\frac{(1+x)^3}{19683}
(x^6+33x^5+204x^4+263x^3+204x^2+33x+1)\nonumber\\
&-\frac{2H}{729}
(3x^6+36x^5+126x^4+166x^3+114x^2+24x+1)\nonumber\\
&+\frac{3x^3+48x^2+48x-5}{81}H^2-\frac8{27}H^3,
\end{align}
\begin{align}
A_2={}&\frac{x(1+x)^6}{59049}(x^4+13x^3+15x^2+13x+1)\nonumber\\
&-\frac{2x(1+x)^2}{19683}
(x^6+25x^5+138x^4+241x^3+226x^2+99x+9)H\nonumber\\
&+\frac{x}{2187}(12x^5+27x^4+243x^3+422x^2+171x-45)H^2\nonumber\\
&+\frac{2x(23x^2-12x-12)}{243}H^3,
\end{align}
\begin{align}
A_1={}&\frac{x^2(1+x)^9(x^2+x+1)}{531441}\nonumber\\
&-\frac{2x^2(1+x)^5}{177147}(2x^4+19x^3+29x^2+27x+9)H\nonumber\\
&+\frac{x^2(1+x)}{59049}
(4x^6+68x^5+85x^4+729x^3+588x^2+513x-135)H^2\nonumber\\
&-\frac{4x^2}{6561}(2x^4-69x^3-51x^2+45x+18)H^3
-\frac{4x^3}{27}H^4,
\end{align}
\begin{align}
A_0={}&-\frac{2x^3(1+x)^8(x^2+x+1)}{1594323}H\nonumber\\
&+\frac{x^3(1+x)^4}{531441}(4x^4+24x^3+57x^2+57x+33)H^2\nonumber\\
&-\frac{2x^3}{531441}
(4x^6+36x^5-585x^4+324x^3-135x^2+486x-486)H^3\nonumber\\
&-\frac{x^3(36x^2+36x-31)}{2187}H^4
+\frac{8x^3}{243}H^5.
\end{align}

\section{Computational verification and ancillary data}

For the algebraic cases the verification is independent of finite-order fitting.  If $P(x,H,F)=0$, implicit differentiation gives
\begin{equation}
 F_x=-\frac{P_x}{P_F},
 \qquad
 F_H=-\frac{P_H}{P_F}.
\end{equation}
Repeated total differentiation expresses every derivative appearing in the differential equation as a rational function of $(x,H,F)$.  After clearing denominators, the numerator is divided by $P$ as a polynomial in $F$.  In the exact $p=4$ and $p=5$ cases the remainder vanishes identically.  The same method verifies the exact $p=6$ specialization.

The large $p=5$ octic, the degree-ten $p=6$ specialization, the factors $\Bcal_5,\Bcal_6$ on the tested slices, and the modular Galois-group checks are best supplied as machine-readable ancillary files rather than printed as multi-page coefficient lists.

\section{Boundary and singularity refinements of the uniform cover}
\label{app:uniform-cover-refinements}

\begin{corollary}[Universal boundary factorization]
\label{cor:uniform-boundary}
With $\beta_q$ and $\beta_{q,\zeta}$ as defined in
Eqs.~\eqref{eq:betaq}--\eqref{eq:betazeta}, the monic minimal
polynomial of Theorem~\ref{thm:uniform-reduced-cover} satisfies
\begin{equation}
 \quad P_p(x,0,Z)=
 Z[Z-\beta_q(x)]^q
 \prod_{\substack{\zeta^q=1\\\zeta\ne1}}
 [Z-\beta_{q,\zeta}(x)].\quad
 \label{eq:uniform-boundary-factorization}
\end{equation}
In particular, the former universal boundary factorization
conjecture holds for every $p\geq3$.
\end{corollary}

\begin{proof}
At $H=0$, Eq.~\eqref{eq:uniform-s-cover} reads
$(1-s^q)(x+s)^q=0$.  Its roots are $s=\zeta$ with $\zeta^q=1$,
each simple, and $s=-x$ with multiplicity $q$.  Evaluation of
Eq.~\eqref{eq:uniform-F-cover} gives
\begin{equation}
 F_p(x,\zeta)=\frac{\varepsilon}{q^3}
 \bigl[(1+x)^q-(1+\zeta^{-1}x)^q\bigr],\qquad
 F_p(x,-x)=\frac{\varepsilon}{q^3}(1+x)^q.
 \label{eq:uniform-boundary-values}
\end{equation}
The first set of values is
$\{\beta_{q,\zeta}(x):\zeta^q=1\}$, with
$\beta_{q,1}=0$; the second value is $\beta_q(x)$.
The norm in Eq.~\eqref{eq:uniform-norm-polynomial} is monic
and specializes regularly at $H=0$ for generic $x$, because
none of the roots $s=\zeta,-x$ is a pole of $F_p(x,s)$ there.
Their multiplicities give Eq.~\eqref{eq:uniform-boundary-factorization}.
\end{proof}

\begin{remark}
The construction also identifies the characteristic critical
values without using the expanded polynomial.  Besides the
multiple root $s=-x$ of $H(s)=0$, the nonzero finite critical points of
$s\mapsto H(s)$ satisfy $s^{q+1}=-x$.  On writing $a=s^q$,
their images obey
$x^q=\varepsilon a^{q+1}$ and
$q^2H=\varepsilon(1-a)^{q+1}$, precisely the parametrization
of $\mathcal A_p(x^q,q^2H)=0$.

Here is a local proof of the cubic characteristic exponent.  Set
$v(s)=\partial_HF_p=(q-1)x/[q(x+s)]$ and let
$s_c^{q+1}=-x$, excluding the nongeneric cases
$s_c=0$ and $x+s_c=0$.  Equation~\eqref{eq:uniform-H-map} gives
\begin{equation}
 H'(s_c)=0,\qquad
 H''(s_c)=-\frac{q+1}{\varepsilon q}
     \frac{(x+s_c)^{q-1}}{s_c}\ne0,\qquad
 v'(s_c)=-\frac{q-1}{q}\frac{x}{(x+s_c)^2}\ne0.
 \label{eq:uniform-cusp-derivatives}
\end{equation}
Since $F_p'(s)=v(s)H'(s)$, with $t=s-s_c$ one finds
\begin{equation}
 \begin{split}
 H(s)-H_c&=\tfrac12H''(s_c)t^2+O(t^3),\\
 F_p(s)-F_c-v(s_c)[H(s)-H_c]
   &=\tfrac13v'(s_c)H''(s_c)t^3+O(t^4).
 \end{split}
 \label{eq:uniform-cusp-local}
\end{equation}
Thus the parametrized plane curve has a $(2,3)$ cusp at a
generic characteristic point.  Its two local $F_p$ branches
over $H$ differ by a nonzero multiple of $(H-H_c)^{3/2}$;
their squared difference has order three.  This establishes
the local cubic discriminant exponent.  Determining every other
factor of the generic discriminant remains a separate question.
\end{remark}

\begin{remark}[A genus budget for the remaining factor]
For generic fixed $x$, Theorem~\ref{thm:uniform-reduced-cover}
identifies the normalization of the image curve
$s\mapsto(H(s),F_p(s))$ with $\mathbb P^1_s$.  Its bidegree in
$\mathbb P^1_H\times\mathbb P^1_F$ is $(2q,2q-1)$, so its
arithmetic genus is $(2q-1)(2q-2)$.  The $q+1$ characteristic
cusps of the preceding remark each have $\delta=1$.

The two parameter values $s=0$ and $s=\infty$ map to the same
point $(H,F)=(\infty,\infty)$.  In local coordinates
$U=1/H$, $V=1/F_p$, the branch at $s=0$ has
\begin{equation}
 \operatorname{ord}_sU=\operatorname{ord}_sV=q,\qquad
 \operatorname{ord}_s\!\left(V-\frac{q}{q-1}U\right)=q+1.
 \label{eq:uniform-infinity-branch-zero}
\end{equation}
The branch at $s=\infty$, with $t=1/s$, has orders
$(\operatorname{ord}_tU,\operatorname{ord}_tV)=(q,q-1)$.
The tangent lines of the two branches are distinct.  Their
individual delta invariants are $q(q-1)/2$ and
$(q-1)(q-2)/2$, while their intersection multiplicity is
$q(q-1)$.  The total delta invariant at this infinity point
is therefore $(q-1)(2q-1)$.

At $s=-x$ one has local orders
$(\operatorname{ord}(H),\operatorname{ord}(F_p-F_p(-x)))
=(q,q-1)$, contributing $(q-1)(q-2)/2$.  After subtracting
these contributions and the $q+1$ characteristic cusps from
the arithmetic genus, the sum of delta invariants at all
remaining points is
\begin{equation}
 (2q-1)(2q-2)-(q+1)-\frac{(q-1)(5q-4)}2
 =\frac{(q-2)(3q+1)}2.
 \label{eq:uniform-remaining-delta}
\end{equation}
If the remaining points are ordinary nodes with distinct
$H$-coordinates, their number agrees with the conjectured
$\deg_H\Bcal_p$ in Eq.~\eqref{eq:Bdegree-conj}.  Establishing
that generic nodality and the complete discriminant factorization
requires an additional argument.
\end{remark}

\begin{remark}[A seventh-order check and exceptional specializations]
The generic discriminant conjecture can be tested beyond the
explicit equations in the main text by specializing the norm
\eqref{eq:uniform-norm-polynomial}.  An exact calculation at
$(p,q,x)=(7,6,3)$ gives, up to a nonzero rational constant,
\begin{equation}
 \operatorname{Disc}_Z P_7(3,H,Z)
 \doteq H^{25}\Acal_7(3^6,36H)^3 C_{7,3}(H)^2,
 \qquad \deg_H C_{7,3}=38.
 \label{eq:uniform-p7-specialization}
\end{equation}
The residual polynomial $C_{7,3}$ is irreducible over
$\mathbb Q$.  This is an exact specialization check of the
predicted degree \eqref{eq:Bdegree-conj}, not a proof of the
factorization for arbitrary $p$.  In characteristic zero its
irreducibility implies squarefreeness; the interpretation of its
roots as ordinary nodes also uses the ramification analysis above.

The qualifier \emph{generic $x$} is essential.  At $q=6$ and
$x=2$, let $\zeta$ be a primitive cube root of unity, which is
also a sixth root.  Then
\begin{equation}
 (1+2\zeta)^6=(1+2\zeta^2)^6=-27,
 \qquad
 \beta_{6,\zeta}(2)=\beta_{6,\zeta^2}(2)=\frac72.
 \label{eq:uniform-exceptional-x2}
\end{equation}
The points $s=\zeta,\zeta^2$ have $H=0$ and the same value of
$F_p$, but are distinct unramified branches: $H'(s)\ne0$ there,
and their slopes
$dF_p/dH=(q-1)x/[q(x+s)]$ differ.  They therefore
form an ordinary node over $H=0$.  Its contribution $H^2$ to
the discriminant moves from the residual factor to the boundary
factor.  Accordingly the specialized calculation has
$H^{27}$ and a remaining squared factor of degree $37$.
The total degree in $H$ is $122$ on both slices:
$25+3\cdot7+2\cdot38=27+3\cdot7+2\cdot37$.
This boundary collision does not contradict the generic
factorization.  Similarly, at $x=1$ the values
$\beta_{q,\zeta}(1)$ and $\beta_{q,\zeta^{-1}}(1)$ coincide
for every $q$-th root $\zeta$.
\end{remark}

\end{document}